\documentclass{article}

\usepackage{graphicx}
\usepackage{amsmath}

\usepackage{amssymb}

\usepackage{amsfonts}
\usepackage{amssymb}
\usepackage{nccmath}
\usepackage{makecell}
\usepackage{braket}
\usepackage{bussproofs}

\newcommand{\CS}{\mathsf{CS}}
\newcommand{\APP}{\mathsf{APP}}
\newcommand{\Tm}{\mathsf{Tm}}

\newcommand{\Prop}{\mathsf{Prop}}

\newcommand{\LJ}{\mathcal{L}_J}

\newcommand{\cstar}{\mathsf{c}^\star}

\newcommand{\model}{\mathcal{M}}

\newcommand{\WnullC}{W_0^C}

\newcommand{\axjplus}{\textbf{j+}}

\newcommand{\axj}{\textbf{j}}

\newcommand{\axjd}{\textbf{jd}}

\newcommand{\axcl}{\textbf{cl}}

\makeatletter
\newcommand{\xLongLeftRightArrow}[2][]{\ext@arrow0055{\LongLeftRightArrowfill@}{#1}{#2}}
\def\LongLeftRightArrowfill@{\arrowfill@\Leftarrow\Relbar\Rightarrow}
\makeatother
\makeatletter
\def\th@plain{%
	\thm@notefont{}% same as heading font
	\itshape % body font
}

\def\th@Definition{%
	\thm@notefont{}% same as heading font
	\normalfont % body font
}
\makeatother

\newcommand{\JD}{\mathsf{JD}}
\newcommand{\JDCS}{\mathsf{JD}_{\CS}}
\newcommand{\Wb}{W_{\not\bot}}

\newcommand{\axnc}{\textbf{noc}}

\newcommand{\JNC}{\mathsf{JNoC}}
\newcommand{\JNCCS}{\mathsf{JNoC}_{\CS}}
\newcommand{\Wnc}{W_{\mathsf{nc}}}

\newcommand{\dab}{D-arbitrary\ }
\newcommand{\noc}{NoC\ }

\usepackage{amsthm}
\theoremstyle{plain}
\newtheorem{theorem}{Theorem}
\newtheorem{lemma}[theorem]{Lemma}
\newtheorem{corollary}[theorem]{Corollary}
\newtheorem{remark}[theorem]{Remark}

\theoremstyle{definition}
\newtheorem{definition}[theorem]{Definition}

\begin{document}

\title{Impossible and Conflicting Obligations in Justification Logic}
\author{Federico L.~G.~Faroldi\thanks{Supported by FWO and FWF Lise Meitner grant M 25-27 G32.} \and
Meghdad Ghari\thanks{This research was in part supported by a grant from IPM (No.~99030420).} \and
Eveline Lehmann\thanks{Supported by the Swiss National Science Foundation grant 200020\_184625.} \and
Thomas Studer\thanks{Supported by the Swiss National Science Foundation grant 200020\_184625.}
}
\date{}

%  \author{First Author}\footnote{You can put your email address or grant acknowledgement as a footnote, if you wish.}
%  \address{Affiliation \\ Address \\ Address}
%  \author{Second Author}
%  \address{Affiliation \\ Address \\ Address}

\maketitle

  \begin{abstract}
Different notions of the consistency of obligations collapse in standard deontic logic. 
In justification logics, which feature explicit reasons for obligations, the situation is different. 
Their strength depends on a constant specification and on the available set of operations for combining different reasons.
We present different consistency principles in justification logic and compare their logical strength. Further, we propose a novel semantics for which justification logics with the explicit version of axiom D, $\axjd$, are complete for arbitrary constant specifications. We then discuss the philosophical implications with regard to some deontic paradoxes.

  \end{abstract}

\section{Introduction}

Deontic logic is the logic of obligations, permissions, and sometimes other (primitive or derived) normative notions. What has emerged as the benchmark version, a system called Standard Deontic Logic (SDL), is nothing more than \textbf{KD}, the smallest normal modal logic with the D axiom schema added. For an introduction and historical overview, see \cite{HilpinenMcNamara13}. 

The D axiom is in place to ensure the consistency of obligations, but can take different formulations, for instance $\neg\mathcal{O}\bot$, or $\mathcal{O}A \rightarrow \neg \mathcal{O}\neg A$, or $\neg(\mathcal{O}A \land \mathcal{O}\neg A)$. In a normal modal logic, all these formulations are provably equivalent, and therefore it does not matter much which one is chosen. In a non-normal (but still classical) setting, for instance when an aggregation principle is missing, different versions of D are not interderivable, and it therefore matters which one is chosen, both for philosophical and for technical reasons (for deontic logic in a paraconsistent setting, see e.g. \cite{da1986paraconsistent}). One might want to distinguish, conceptually, between an obligation for an impossible or logically contradictory state of affairs ($\mathcal{O}\bot$) on one hand, and multiple obligations for jointly inconsistent states of affairs ($\mathcal{O}A \land \mathcal{O}\neg A$) on the other, because the former might thought to be self-defeating or conceptually impossible, whereas the latter can derive from different background or contingent normative systems (e.g. ethics and the law) and are only practically unenforceable, but logically possible. 
Moreover, SDL and its variants in the standard modal language lack the power to distinguish the source of one situation (an obligation for the impossible) from the source of the other (inconsistent obligations), or to exclude one situation for logical reasons and admit the other for contingent reasons. 

Justification logic~\cite{artemovFittingBook,justificationLogic2019} is an explicit version of modal logic originally developed to provide a logic of proofs~\cite{Art01BSL,KSWeak}. Instead of formulas such as $\Box A$, the language of justification logic includes formulas such as $t:A$ saying, for instance, that \emph{$t$ justifies knowledge of $A$} or \emph{$A$ is obligatory because of reason $t$}, where $t$ is a term representing the reason. 
Systems of justification logic are parameterized by a so-called \emph{constant specification} that states which logical axioms do have a justification. Hence the constant specification can be used to calibrate the strength of a justification logic. Of particular interest are axiomatically appropriate constant specifications where every axiom has a justification. In that case the justification logic enjoys a constructive analogue of the modal necessitation rule. (See Sect. 3 for a formal definition of constant specification).

%%%%%%
The explicit counterpart in justification logic of one version of D (in standard modal logic) was first formulated by Brezhnev~\cite{Bre00TR} as axiom $\axjd$, i.e.~$\lnot (t:\bot)$. This axiom turned out to be rather notorious. Usually one can establish completeness of a justification logic for an arbitrary constant specification. However, in the presence of $\axjd$ this is not the case. Systems that include $\axjd$ usually need an axiomatically appropriate constant specification in order to be complete. 
Kuznets~\cite{Kuz00CSL} defined M-models for justification logics with $\axjd$  and 
Pacuit~\cite{Pac05PLS} presented F-models for $\axjd$.
Modular models for  $\axjd$ have been studied in~\cite{KuzStu12AiML} and subset models for  $\axjd$ are introduced in~\cite{StuderLehmannSubsetModel2019}.
For all these different semantics, an axiomatically appropriate constant specification is required in order to obtain a completeness result.
Notable exceptions to this phenomenon are M-models (defined in~\cite{Kuz00CSL}) and Fk-models (defined in~\cite{Kuz08PhD}) for which completeness holds for arbitrary constant specifications. 

The requirement of an axiomatically appropriate constant specification is often overlooked. 
In particular, this requirement is omitted in the completeness theorems given in 
\cite{FaroldiHyperPracticalReasons} and
\cite{StuderLehmannSubsetModel2019} (although for the latter paper it seems that it has been corrected later~\cite{LehmannStuderArXiv19}).
In the case of~\cite{FaroldiHyperPracticalReasons} the appropriateness requirement is important since the solution to avoid some of the known paradoxes, such as Ross', is to restrict the constant specification. But then the resulting justification logic is not complete anymore.

In this paper we study in detail 
the  $\axjd$ axiom and related principles. We compare their logical strength and we investigate the role of the constant specification. After an informal discussion in Section 2, in Sections 3 and 4 we present the basic syntax and semantics of system $\JD$ (with axiom $\axjd$). In Section 5 we propose a novel semantics for which justification logics with $\axjd$ are complete for arbitrary constant specifications. In Section 6, we consider system $\JNC$, which has a different version of the consistency axiom, \textbf{noc}. In Section~7, we establish that various formulations of consistency are equivalent only with an axiomatically appropriate constant specification.

{\bf Acknowledgements.} We are grateful to the anonymous reviewers of DEON 2020 for many helpful comments.

\section{Impossible vs Inconsistent Obligations: An Overview}\label{overview}

Standard (implicit) systems of deontic logic conflate impossible and conflicting obligations. One thing is to say that nothing logically impossible can be obligatory, i.e.~$\lnot \mathcal{O} \bot$, another to say that there are not (or there should not be) conflicting provisions that are obligatory, i.e.~$\lnot(\mathcal{O} A \land \mathcal{O} \neg A)$. Standard systems can derive $\mathcal{O} \bot$ from $\mathcal{O} A \land \mathcal{O} \neg A$ and vice versa, thus suffering a collapse.
One way to see the difference is that the former might be argued to be unacceptable for conceptual or logical reasons (e.g. that such an obligation would be conceptually self-defeating), whereas the latter might be argued to be unacceptable for contingent reasons (e.g. that such obligations cannot be fulfilled in reality, although can potentially still arise in real-life situations). 
\cite{chellas74,chellas} use minimal models, \cite{nonkrikedeon} 
uses multiple accessibility relations in the disjunctive truth condition of the ought operator: in such ways the authors avoid  aggregation and therefore the collapse of impossible to inconsistent obligations (multi-relational semantics has also been used more recently, cf.~e.g.~\cite{CalardoForthcoming-CALQIS}).

In justification logic we have the explicit counterparts \textbf{jd}: $\lnot (t :\bot)$ and \textbf{noc}: \mbox{$\lnot (t: A \land t:\neg A )$}, respectively, of the above implicit principles, giving rise to systems we call $\JD$ and $\JNC$ (respectively).
Corollary~\ref{c:onedir:1} establishes that the former implies the latter. Lemma~\ref{l:notbot:2}
shows that the converse direction holds in the presence  of an axiomatically appropriate constant specification. In this situation we have the same 
collapse as in the standard implicit systems. There are two options to avoid this consequence:
\begin{itemize}
\item
In justification logic we can use the constant specification to adjust the power of the logical systems and thus avoid the collapse. Lemma~\ref{onedir:3} shows that $\lnot (t: A \land t:\neg A )$ does not imply $\lnot (t :\bot)$ if the constant specification is not axiomatically appropriate. Theorems~\ref{th:sc:1} and~\ref{th:sc:2} prove that $\JD$ and $\JNC$ with an arbitrary constant specification are complete with regard to a novel semantics we develop.
\item
As explained in Remark~\ref{rem:sum:1}, we can avoid the collapse even in the presence of an axiomatically appropriate $\CS$.
It suffices to consider a language without the $+$ operation. 
We denote this system $\JNC^-$.
\end{itemize}

Avoiding this collapse is important in situations with conflicting obligations. Let us look at Sartre's Dilemma~\cite{Lemmon1962-LEMMD} as presented in~\cite{McNamara2019}:
\begin{enumerate}
\item
It is obligatory that I now meet Jones (say, as promised to Jones, my friend).
\item
It is obligatory that I now do not meet Jones (say, as promised to Smith, another friend).
\end{enumerate}

In implicit standard deontic logic featuring the principle $\lnot (\mathcal{O} A \land \mathcal{O} \neg A)$, we immediately get a contradiction if we represent (1) and (2) as $\mathcal{O} A$ and $\mathcal{O}  \lnot A$, respectively. 
However, in a system such as $\JNC^-$, there is no conflict as there are two different reasons in (1) and (2). Hence (1) and (2) are represented by
$s:A$ and $t:\lnot A$ for two different terms $s$ and $t$, which is consistent with axiom~$\axnc$.

Moreover, in normal deontic logic, 
one can pass from \emph{two} inconsistent obligations to \emph{one} impossible obligation. This is dubious on philosophical grounds: we have pointed out that one may argue that one impossible obligation is conceptually self-defeating, whereas two inconsistent obligations may be in place for contingent reasons (e.g. different promises). 

Justification logic gives us the means to not conflate the two, without loosing too much reasoning power. Even more, one can do justice to the background philosophical intuitions to exclude impossible obligations for logical reasons, for instance by focusing on the system $\JNC$ and calibrating the constant specification.
Keeping track of the source of obligations, for instance through reasons, opens up the possibility to solve conflicts if one has a priority ordering on reasons (see for instance \cite{Horty2012}, and \cite{farolditudoATC} for an implementation in justification logic).

In the rest of the paper we present the formal results starting from the basic syntax and semantics of system $\JD$ (with axiom $\axjd$).

\section{Syntax}

Justification terms are built from countably many constants $c_i$ and variables~$x_i$ 
according to the following grammar:
\[
t::=c_i \ |\  x_i\ | \ t \cdot t\ | \ (t+t) \ | \ !t 
\]
The set of terms is denoted by $\Tm$.

Formulas are built from countably many atomic propositions $P_i$ and the  symbol~$\perp$ according to the following grammar:
\[
F::=P_i \ | \ \perp \ | \ F\to F \ | \ t:F
\]
The set of atomic propositions is denoted by $\Prop$ and the set of all formulas is denoted by $\LJ$. The other classical Boolean connectives $\neg,\top, \land,\lor,\leftrightarrow$ are defined as usual, in particular we have $\lnot A := A \to \bot$ and $\top := \lnot \bot$. Informally, $+$ mimics the aggregation of reasons, $\cdot$ embodies modus ponens reasoning, and ! is positive introspection. We keep ! for ease of exposition, but it can be dispensed with. For a discussion on the interpretation of the operations in a deontic context, see~\cite{FaroldiHyperPracticalReasons}.

The axioms of $\JD$ are the following:
\begin{fleqn}
\begin{equation}
\begin{array}{ll}\nonumber
\axcl & \text{all axioms of classical propositional logic};\\
\axjplus & s:A\lor t:A\to (s+t): A;\\
\axj & s:(A\to B)\to (t:A\to s\cdot t:B);\\
\axjd & \lnot (t:\perp).
\end{array}
\end{equation}
\end{fleqn}
Note that since $\lnot$ is a defined notion, $\axjd$ actually stands for $t:\perp\to\perp$.

Justification logics are parameterized by a so-called constant specification, which is a set
\[
\CS\subseteq\{(c, A)\enspace|\enspace c\text{ is a constant and } A\text{ is an axiom of }\JD\}.
\]
Our logic $\JDCS$ is now given by the axioms of $\JD$ and the rules modus ponens:
\begin{prooftree}
	\AxiomC{$A$}
	\AxiomC{$A\to B$}
	\RightLabel{(MP)}
	\BinaryInfC{$B$}
\end{prooftree}
and axiom necessitation

\begin{prooftree}
	\AxiomC{}
	\RightLabel{(AN!)\quad $\forall n\in\mathbb{N}$, where $(c, A)\in\CS$}
	\UnaryInfC{$\underbrace{!...!}_{n}c:\underbrace{!...!}_{n-1}c:\ ...:\ !!c:\ !c:c:A$}
\end{prooftree}

\begin{definition}[Axiomatically appropriate $\CS$]\label{def:cs_axiomatocally_appropriate}
	A constant specification $\CS$ is called \emph{axiomatically appropriate} if for each axiom $A$, there is a constant $c$ with $(c,A) \in \CS$.
\end{definition}

Axiomatically appropriate constant specifications are important as they provide a form of necessitation~\cite{Art01BSL,artemovFittingBook,justificationLogic2019}.

\begin{lemma}
Let $\CS$ be an axiomatically appropriate constant specification.
For each formula $A$ with 
\[
\JDCS \vdash A,
\]
 there exists a term $t$ such that
\[
\JDCS \vdash t:A.
\]
\end{lemma}

\section{Semantics}

We recall the basic definitions and results about subset models for justification logic~\cite{StuderLehmannSubsetModel2019,StuderLehmannSubsetModelLFCS2020,exploringSM}.

\begin{definition}[General subset model] Given some constant specification~$\CS$, then a general $\CS$-subset model $\mathcal{M}=(W, W_0, V, E)$ is defined by:
\begin{itemize}
\item $W$ is a set of objects called worlds.
\item $W_0\subseteq W$ and $W_0\neq\emptyset$ .
\item $V: W\times\LJ\to \{0, 1\}$  such that for all $\omega\in W_0$, $t\in\Tm$, $F, G\in\LJ$:
	\begin{itemize}
	\item $V(\omega, \perp)=0$;
	\item $V(\omega, F\to G)=1\quad\text{ iff }\quad V(\omega,F)=0$ or $V(\omega, G)=1$;
	\item\label{condition_tF} $V(\omega, t:F)=1\quad\text{ iff }\quad E(\omega, t)\subseteq\Set{\upsilon\in W | V(\upsilon, F)=1}$.
	\end{itemize}
\item $E: W\times\Tm \to \mathcal{P}(W)$ that meets the following conditions
where we use 
\begin{equation}\label{eq:truthset:1}
[A]:=\{\omega\in W\enspace|\enspace V(\omega, A)=1\}.
\end{equation}
For all $\omega\in W_0$, and for all $s, t\in\Tm$:
	\begin{itemize}
	\item $E(\omega, s+t)\subseteq E(\omega, s)\cap E(\omega, t)$;
	\item $E(\omega, s\cdot t)\subseteq\{\upsilon\in W\enspace|\enspace \forall F\in\APP_\omega(s, t)(\upsilon\in[F])\}$ where $\APP$ contains all formulas that can be justified by an application of~$s$ to~$t$, see below;
	\item $\exists\upsilon\in W_0$ with $\upsilon\in E(\omega, t)$;
		\item for all $n\in\mathbb{N}$ and for all $(c, A)\in\CS: E(\omega, c)\subseteq[A]$ and 
	\[E(\omega, \underbrace{!...!}_{n}c)\subseteq[\underbrace{!...!}_{n-1}c:....!c:c:A].\]

	\end{itemize}
\end{itemize}
The set $\APP$ is formally defined as follows:
\begin{align*}
\APP_\omega(s, t):=\{F\in\LJ \ |\enspace	&\exists H\in\LJ \text{ s.t.~}\\
									&E(\omega, s)\subseteq[H\to F]\text{ and }E(\omega, t)\subseteq[H]\};
\end{align*}
\end{definition}
$W_0$ is the set of \emph{normal} worlds. 
The set $W \setminus W_0$ consists of the \emph{non-normal} worlds. 
Moreover, using the notation introduced by \eqref{eq:truthset:1}, we can read the condition on $V$ for justification formulas $t:F$ as:
\[V(\omega, t:F)=1\quad\text{ iff }\quad E(\omega, t)\subseteq[F]\]

In subset semantics terms are not treated only syntactically (as in most other semantics for justification logics), but they get assigned a set of worlds. 

$E(\omega, t)$ tells us the states that are ideal according to $t$ from $\omega$'s perspective.  
 Then $t : F$ at $\omega$ is true just in case $F$ is true at those ideal states. 
We have seen that a formula of the form $t : F$ is true at a world $w$ just in case the interpretation of $t$ at $w$ (a set of worlds) is a subset of the truth set of $F$ (the set of worlds where $F$ is true). However, take two axioms $A$ and $B$. They are true in all possible worlds. Therefore, every term that is a reason for the former will also be a reason for the latter (if terms get assigned sets of \emph{possible} worlds). But in this way, there is no control on the constant specification. Using \emph{impossible} worlds, however, lets us solve this problem, because at impossible worlds classical logically equivalent propositions can differ in truth value, and a justification for one may not be a justification for the other. This makes the semantics able to capture hyperintensionality.

Since the valuation function $V$ is defined on worlds and formulas, the definition of truth is standard.

\begin{definition}[Truth] Given a subset model 
\[
\mathcal{M}=(W, W_0, V, E)
\]
and a world\/ $\omega\in W$ and a formula $F$ we define the relation $\Vdash$ as follows:
\[\mathcal{M}, \omega\Vdash F\quad\text{ iff }\quad V(\omega, F)=1.\]
\end{definition}

Validity is defined with respect to the normal worlds.
\begin{definition}[Validity]
Let\/ $\CS$ be a constant specification.
We say that a formula $F$ is \emph{general $\CS$-valid} if for each general $\CS$-subset model 
\[
\mathcal{M}=(W, W_0, V, E)
\]
and each $\omega\in W_0$, we have $\mathcal{M}, \omega\Vdash F$.
\end{definition}

As expected, we have soundness~\cite{StuderLehmannSubsetModel2019}.
\begin{theorem}[Soundness] Let $\CS$ be an arbitrary constant specification. 
For each formula $F$ we have that if $\JDCS\vdash F$, then $F$ is general $\CS$-valid.
\end{theorem}

However, completeness only holds if the constant specification is axiomatically appropriate~\cite{LehmannStuderArXiv19}.
\begin{theorem}[Completeness] Let $\CS$ be an axiomatically appropriate constant specification. 
For each formula $F$ we have that if $F$ is general $\CS$-valid, then $\JDCS\vdash F$.
\end{theorem}

One might need more control on the constant specification, e.g.~by relinquishing the requirement that each axiom be justified. For instance,~\cite{FaroldiHyperPracticalReasons} argued that restricting the constant specification is one way to avoid certain deontic paradoxes, such as Ross'. In the next section, we prove soundness and completeness with regard to an arbitrary constant specification.

\section{\dab subset models}

We present a novel class of subset models for $\JD$ and establish soundness and completeness.

\begin{definition}[\dab subset model]
A \emph{\dab $\CS$-subset model} $\mathcal{M}=(W, W_0, V, E)$ is defined like a general $\CS$-subset model with the condition
\[
\exists\upsilon\in W_0 \text{ with }\upsilon\in E(\omega, t)
\]
being replaced with 
\[
\exists\upsilon\in \Wb \text{ with }\upsilon\in E(\omega, t)
\]
where $\Wb:=\{\omega\in W\ | \ V(\omega, \bot)=0\}$.
\end{definition}

The notion of \dab $\CS$-validity is now as expected.
\begin{definition}[\dab validity]
Let $\CS$ be a constant specification.
We say that a formula $F$ is \emph{\dab $\CS$-valid} if for each  \dab $\CS$-subset model $\mathcal{M}=(W, W_0, V, E)$ and each $\omega\in W_0$, we have $\mathcal{M}, \omega\Vdash F$.
\end{definition}

We have soundness and completeness with respect to arbitrary constant specifications.

\begin{theorem}[Soundness and Completeness]\label{th:sc:1} 
Let $\CS$ be an arbitrary constant specification. 
For each formula $F$ we have 
\[
\JDCS\vdash F \quad\text{if{f}}\quad \text{$F$ is D-arbitrary $\CS$-valid.}
\]
\end{theorem}
The completeness proof is by a canonical model construction as in the case of general subset models~\cite{StuderLehmannSubsetModel2019}.
We will only sketch main steps here. The canonical model is given as follows.

\begin{definition}[Canonical Model]\label{d:canMod:1} Let $\CS$ be an arbitrary constant specification. 
We define the canonical model $\model^C=(W^C, \WnullC, V^C, E^C)$ by:
\begin{itemize}
\item $W^C=\mathcal{P}(\LJ)$.
\item $\WnullC=\Set{\Gamma\in W^C| \Gamma\text{ is maximal $\JDCS$-consistent set of formulas}}$.
\item $V^C(\Gamma, F)=1\quad\text{ iff }\quad F\in\Gamma$;
\item $E^C(\Gamma, t)=\Set{\Delta\in W^C|\Delta\supseteq \Gamma/t}$  where
\[
\Gamma/t:=\{F\in\LJ\enspace|\enspace t:F\in\Gamma\}.
\]
\end{itemize}
\end{definition}

The essential part of the completeness proof is to show that the canonical model is a \dab $\CS$-subset model.
\begin{lemma}\label{l:isModel:1}
Let $\CS$ be an arbitrary constant specification.
The canonical model $\model^C$ is a \dab $\CS$-subset model.
\end{lemma}
\begin{proof}
Let us only show the condition 
\begin{equation}\label{eq:D:1}
\exists\upsilon\in \Wb^C \text{ with }\upsilon\in E(\omega, t)
\end{equation}
for all $\omega \in W_0$ and all terms $t$.

So let $t$ be an arbitrary term and $\Gamma \in \WnullC$. Since $\Gamma$ is a maximal $\JDCS$-consistent set of formulas, we find
$\lnot (t:\bot) \in \Gamma$ and thus $t:\bot \notin \Gamma$. Let $\Delta := \Gamma/t$. We find that $\bot\notin\Delta$ and by definition $V^C(\Delta,\bot)=0$. Thus $\Delta \in \Wb^C$. Moreover, again by definition, $\Delta \in E^C(\Gamma,t)$. Thus~\eqref{eq:D:1} is established.
 \end{proof}

Now the Truth lemma and the completeness theorem follow easily as in~\cite{StuderLehmannSubsetModel2019}.

\begin{remark}
In subset models, it is possible to reduce application to sum by introducing a new term $\cstar$, see~\cite{StuderLehmannSubsetModel2019}.
Our completeness result also holds in the setting with $\cstar$. However, the proof that the canonical model is well-defined is a bit more complicated as one has to consider the case of $\cstar$ separately.
\end{remark}

\section{No conflicts}
So far, we have considered the explicit version of $\neg\mathcal{O}\bot$. In normal modal logic, this is provably equivalent to $\neg(\mathcal{O}A \land \mathcal{O}\neg A)$.
In this section we study the explicit version of this principle, which we call \noc 
(\emph{No Conflicts}), saying that reasons are self-consistent. That is $A$ and $\lnot A$ cannot be obligatory for one and the same reason.
The axioms of $\JNC$ are the axioms of $\JD$ where $\axjd$ is replaced with:
\begin{fleqn}
\begin{equation}
\begin{array}{ll}\nonumber
\axnc & \lnot (t:A \land t:\lnot A).
\end{array}
\end{equation}
\end{fleqn}

Accordingly, a constant specification for $\JNC$ is defined like a constant specification for $\JD$ except that the constants justify axioms of $\JNC$.

Given a constant specification $\CS$ for  $\JNC$, the logic $\JNCCS$ is defined by the axioms of $\JNC$ and the rules modus ponens and axiom necessitation.

\begin{definition}[\noc subset model]
A \emph{\noc $\CS$-subset model} 
\[
\mathcal{M}=(W, W_0, V, E)
\]
is defined like a general $\CS$-subset model with the condition
\[
\exists\upsilon\in W_0 \text{ with }\upsilon\in E(\omega, t)
\]
being replaced with 
\[
\exists\upsilon\in \Wnc \text{ with }\upsilon\in E(\omega, t)
\]
where $\Wnc:=\{\omega\in W\ | \ \text{for all formulas A } (V(\omega, A) =0 \text{ or } V(\omega, \lnot A) =0)  \}$.
\end{definition}

The notion of \noc $\CS$-validity is now as expected.
\begin{definition}[\noc validity]
Let $\CS$ be a constant specification.
We say that a formula $F$ is \emph{ \noc $\CS$-valid} if for each   \noc $\CS$-subset model $\mathcal{M}=(W, W_0, V, E)$ and each $\omega\in W_0$, we have $\mathcal{M}, \omega\Vdash F$.
\end{definition}

\begin{theorem}[Soundness and Completeness]\label{th:sc:2}
Let $\CS$ be an arbitrary constant specification. 
For each formula $F$ we have 
\[
\JNCCS\vdash F \quad\text{if{f}}\quad \text{$F$ is \noc $\CS$-valid.}
\]
\end{theorem}

Again the completeness proof uses the canonical model construction from Definition~\ref{d:canMod:1} except that we set
\begin{itemize}
\item $\WnullC=\Set{\Gamma\in W^C| \Gamma\text{ is maximal $\JNCCS$-consistent set of formulas}}$.
\end{itemize}
 Now we have to show that the defined structure is an \noc $\CS$-subset model.

\begin{lemma}
Let $\CS$ be an arbitrary constant specification.
The canonical model $\model^C$ is an  \noc $\CS$-subset model.
\end{lemma}
\begin{proof}
As before, we only show the condition 
\begin{equation}\label{eq:D:2}
\exists\upsilon\in \Wnc^C \text{ with }\upsilon\in E(\omega, t)
\end{equation}
for all $\omega \in W_0$ and all terms $t$.

So let $t$ be an arbitrary term  and $\Gamma \in \WnullC$.
Let $A$ be an arbitrary formula. 
Since $\Gamma$ is a is maximal $\JNCCS$-consistent set of formulas, we find
\[
\lnot (t:A \land t:\lnot A) \in \Gamma
\]
and thus $t:A \land t:\lnot A \notin \Gamma$. Thus, again by maximal consistency, 
\[
t:A \notin \Gamma \text{ or } t:\lnot A \notin \Gamma.
\]
Let $\Delta := \Gamma/t$. We find that 
\[
A \notin \Delta \text{ or } \lnot A \notin \Delta
\]
and hence, by definition,
\[
V^C(\Delta,A)=0 \text{ or } V^C(\Delta,\lnot A)=0.
\]
Thus $\Delta \in \Wnc^C$. Moreover, again by definition, $\Delta \in E^C(\Gamma,t)$. Thus~\eqref{eq:D:2} is established.
\end{proof}

Again the Truth lemma and the completeness theorem follow easily.

\section{Formal comparison}

The following lemmas establish the exact relationship between $\JD$ and $\JNC$.
First we show that $\JDCS$ proves that reasons are consistent among them, i.e.~that $\lnot (s:A \land t:\lnot A)$ holds for arbitrary terms $s$ and $t$, which is the consistency principle used in~\cite{FaroldiHyperPracticalReasons}.

\begin{lemma}\label{onedir:1}
Let\/ $\CS$ be an arbitrary constant specification.
Then $\JDCS$ proves $\lnot (s:A \land t:\lnot A)$ for all terms $s, t$ and all formulas $A$.
\end{lemma}
\begin{proof}
Suppose towards a contradiction that $s : A \wedge t : \neg A$.
Thus we have $s : A$ and $t:\lnot A$ where the latter is an abbreviation for $ t : (A \rightarrow \bot)$ (by the definition of the symbol $\lnot$).
Thus using axiom $\axj$, we get  $t \cdot s :\bot$ and by axiom~$\axjd$ we conclude $\bot$.
\end{proof}

\begin{corollary}\label{c:onedir:1}
For any constant specification $\CS$,  $\JDCS$ proves every instance of $\axnc$.
\end{corollary}

\begin{remark}
It is only by coincidence that Lemma~\ref{onedir:1}, and thus also Corollary~\ref{c:onedir:1}, hold for arbitrary constant specifications.
If we base our propositional language on different connectives (say $\land$ and $\lnot$ instead of $\to$ and $\bot$), then 
Lemma~\ref{onedir:1} and  Corollary~\ref{c:onedir:1} only hold for axiomatically appropriate constant specifications.

The proof of Lemma~\ref{onedir:1}  is as follows. Since $\CS$ is axiomatically appropriate, there exists a term $r$ such that 
\begin{equation}\label{eq:coinc:1}
r: (\lnot A \to (A \to \bot)) 
\end{equation}
is provable where $\bot$ is defined as $P \land \lnot P$ (for some fixed $P$) and $F \to G$ is defined as $\lnot (F \land \lnot G)$.
From \eqref{eq:coinc:1} and axiom $\axj$ we get
\[
t : \lnot A \to r\cdot t : (A \to \bot).
\]
Thus from  $s : A \wedge t : \neg A$,  we obtain $(r\cdot t) \cdot s: \bot$, which contradicts axiom $\axjd$ as before.
\end{remark}

Next we show that also $\JNCCS$ proves that reasons are consistent among them.

\begin{lemma}\label{l:notbot:1}
Let\/ $\CS$ be an arbitrary constant specification. Then $\JNCCS$ proves $\lnot (s:A \land t:\lnot A)$ for all terms $s, t$ and all formulas $A$.
\end{lemma}
\begin{proof}
Suppose towards a contradiction that $s:A \land t:\lnot A$ holds.
Using axiom~$\axjplus$ we immediately obtain $s+t:A \land s+t:\lnot A$.
By axiom $\axnc$ we conclude $\bot$, which establishes $\lnot (s:A \land t:\lnot A)$.
\end{proof}

Next we show  that $\JNCCS$ with an axiomatically appropriate  constant specification proves $\lnot (t:\bot)$.

\begin{lemma}\label{l:notbot:2}
Let\/ $\CS$ be an axiomatically appropriate constant specification. Then $\JNCCS$ proves $\lnot (t: \bot)$ for each term $t$.
\end{lemma}
\begin{proof}
Because $\CS$ is axiomatically appropriate, there are terms $r$ and $s$ such that
\[
r:(\bot \to P)  \quad\text{and}\quad s:(\bot  \to \lnot P).
\]
Therefore, we get
\[
t:\bot \to r\cdot t: P   \quad\text{and}\quad t:\bot \to s\cdot t: \lnot P.
\]
Thus we have $t:\bot \to (r\cdot t: P \land  s\cdot t: \lnot P)$. Together with the previous lemma, this yields
$t :\bot \to \bot$, which is $\lnot (t :\bot)$.
\end{proof}

Here the requirement of an axiomatically appropriate constant specification is necessary.

\begin{lemma}\label{onedir:3}
There exists  a
\noc $\CS$-subset model $\mathcal{M}=(W, W_0, V, E)$ with some $\omega \in W_0$ such that 
\[
\mathcal{M}, \omega \Vdash t:\bot
\]
for some term $t$.
\end{lemma}
\begin{proof}
Consider the empty $\CS$ and the following model:
\begin{enumerate}
\item $W =\{\omega, \nu\}$ and $W_0 =\{\omega\}$
\item $V(\nu,\bot) =1$ and $V(\nu,F) =0$ for all other formulas $F$
\item $E(\omega ,t) = \{\nu\}$ for all terms $t$.
\end{enumerate} 
We observe that $\nu \in \Wnc$. 
So the model is well-defined.
Further, we find $E(\omega ,t) \subseteq [\bot]$.
Since $\omega \in W_0$, we get $V(\omega, t:\bot)=1$. We conclude
\[
\mathcal{M}, \omega \Vdash t:\bot.
\qedhere
\]
\end{proof}

\begin{remark}\label{rem:sum:1}
For Lemmas~\ref{l:notbot:1} and~\ref{l:notbot:2}, the presence of the $+$ operation is essential. Consider a term language without $+$ and the logic $\JNC^-$ being  $\JNC$ without $\axjplus$. Let $\CS$ be an axiomatically appropriate $\CS$ for $\JNC^-$.
There is a \noc $\CS$-subset model $\mathcal{M}$ for  $\JNCCS^-$ with a normal world $\omega$ such that
\[
\mathcal{M}, \omega \Vdash s:P \land t:\lnot P 
\]
for some terms $s$ and $t$ and some proposition $P$. 

Hence if we drop the $+$ operation, we can have self consistent reasons without getting reasons that are consistent among them even in the presence of an axiomatically appropriate constant specification.
\end{remark}

Instead of using an axiomatically appropriate constant specification, we could also add the schema $s:\top$ to $\JNCCS$ in order to derive $\axjd$. 

\begin{lemma}
Let $\CS$ be an arbitrary constant specification.
Let $\JNCCS^+$ be
$\JNCCS$ extended by the schema
$
s:\top
$
for all terms $s$.
We find that
\[
\JNCCS^+ \vdash \lnot (t: \bot) \quad\text{ for each term $t$.}
\]
\end{lemma}
\begin{proof}
The following is an instance of axiom $\axnc$
\[
\lnot (t: \bot \land t:\lnot \bot).
\]
Using the definition $\top := \lnot \bot$ and propositional reasoning, we obtain
\[
t:\top \to \lnot (t :\bot).
\]
Using $t:\top$ and modus ponens, we conclude $\lnot (t: \bot)$.
\end{proof}

\section{Remarks}
There are two main advantages in using the justification logic framework to deal with deontic matters. First, one can explicitly track which reasons are reasons for what and perform operation on them, thus having a higher degree of accuracy in formal representations of normative reasoning: every obligation has a source. Puzzles and paradoxes such as Ross' are very easy to identify and, under a plausible set-up, disappear. In the present paper we have seen how justification logic provides a means to keep track of the source of impossible and inconsistent obligations, thus helping not to conflate the two.

Second, the framework allows for the hyperintensionality of obligation, namely that logically equivalent contents may not be normatively equivalent. In general it is not the case that if $t : F$ and $F \equiv G$, then $t : G$.
This also ensures a finer-grained formal approach to everyday normative reasoning that is currently unavailable in more standard approaches.

When we come to the specific topic of the present paper, however, we have to remark that it is possible to distinguish between $\neg\mathcal{O}\bot$ and $\neg(\mathcal{O}A \land \mathcal{O}\neg A)$ also in some non-normal implicit modal systems, as we noted in Sect.~\ref{overview}, and in particular in Chellas' system \textbf{D} (cf.~\cite{chellas74,chellas}), which dispenses with axiom schema M: $\mathcal{O}(A \land B) \rightarrow \mathcal{O}A \land \mathcal{O}B$.

Chellas minimal monadic deontic logic \textbf{D} builds as usual on PC, adds $\neg\mathcal{O}\bot$ as an axiom, and has rule ROM: $A \rightarrow B / \mathcal{O}A \rightarrow \mathcal{O}B$.
In Chellas' logic the collapse is indeed avoided, because $\neg(\mathcal{O}A \land \mathcal{O}\neg A)$ is not derivable from $\neg\mathcal{O}\bot$.

How does Chellas' approach compare to the one developed in the present paper? Given the apparent similarities, let's focus on the differences, both technical and philosophical. Rule ROM could be questioned in a deontic context: however, this rule is fundamental in Chellas' system, therefore one cannot ignore it (selectively or not); whereas in a justification logic context we can have a finer-grained control on which axioms get an ``automatic", as it were, normative justification, by fine-tuning the constant specification. 

Philosophically, we can start from the semantic interpretation of the obligation operator. For Chellas, ``OA is true at a possible world just in case the world has a non-empty class of deontic alternatives throughout which A is true. The picture is one of possibly empty collections of non-empty classes of worlds functioning as moral standards: what ought to be true is what is entailed by one of these moral standards \cite[p.24]{chellas74}''. Chellas uses a neighborhood semantics. A standard, for him, is a collection of propositions. A term, in the context of the present paper, is instead interpreted as a set of worlds.

Moreover, Chellas' system is still an implicit modal logic, so it cannot keep track and reason with the sources of obligations. And indeed this reading is consistent with Chellas' intended interpretation of the obligation operator: What ought to be true is what is entailed by one of these moral standards. But which? In a justification logic context, for instance, if one wants to retain Chellas' ideas to interpret terms as moral standards, one can keep track of which moral standard requires what.

\section{Conclusion}
We provided a novel semantics for justification logics with axiom D that does \emph{not} require an axiomatically appropriate constant specification, i.e.~not every axiom needs to be justified by a constant. This can be crucial to have more control on the logic and solve some traditional puzzles such as Ross'.
Axiom~D can be formulated in at least two equivalent ways in normal modal logic, either 
with inconsistent obligations ($\neg(\mathcal{O}A \land \mathcal{O}\neg A)$) 
or
with one impossible obligation ($\neg\mathcal{O}\bot$). 
We proved that 
their explicit versions are interderivable in $\JD$ and $\JNC$ only when the constant specification is axiomatically appropriate.
In particular, our technical results are:
\begin{enumerate}
\item $\JDCS$ proves $\axnc$ for axiomatically appropriate $\CS$ and vice versa
\item $\JNCCS$ proves $\axjd$ for axiomatically appropriate $\CS$.
\item $\JDCS$ proves $\axnc$ for arbitrary $\CS$ only if the language is based on the Boolean  connectives $\to$ and $\bot$.
\item $\JNCCS$ does not prove $\axjd$  for arbitrary $\CS$. 
\item $\JNCCS^-$ does not prove $\axjd$  for axiomatically appropriate  $\CS$. 
\end{enumerate}
Having more control not only on how to formulate D, but also on how to specify the constant specification is philosophically perspicuous: it avoids conflating impossible and conflicting obligations and can encode why this is the case, e.g. for conceptual (logical) or contingent reasons.

Recently, it was shown that principle $\axnc$ is also very useful for analyzing epistemic situations in the context of quantum physics~\cite{StuderQuantum}.

%% Bibliography
%% Make sure to use the bibliographystyle deon16.
%\bibliographystyle{deon16}
%\bibliography{mybib,JLBibliography}

\end{document}